\documentclass[letterpaper]{article} 
\usepackage[preprint]{aaai2027} 
\usepackage[hyphens]{url} 
\usepackage{graphicx} 
\usepackage{natbib} 
\usepackage{caption} 
\usepackage{booktabs}
\usepackage{amsmath,amssymb,amsthm}

\newtheorem{proposition}{Proposition}
\newtheorem{theorem}{Theorem}
\newtheorem{corollary}{Corollary}
\newcommand{\E}{\mathbb{E}}
\newcommand{\R}{\mathbb{R}}
\newcommand{\Var}{\operatorname{Var}}

\title{When Unpaired Sets Support Shared-Corruption Calibration: Moment Geometry and Two-Sample Precision}
\author{
Shuheng Cao$^{1,*}$,
Zhenhao Zhang$^{4,6,*}$,
Ruiqi Chen$^{2,*}$,
Renjie Cao$^{3,\dagger}$,
Siyu Zhang$^{1,\dagger}$,
Zhaoxiang Feng$^{1,\dagger}$,
Lingwei Dang$^{5,\dagger}$,
Haoyang Wu$^{7,\dagger}$
}

\affiliations{
\begin{tabular}{@{}c@{}}
$^{1}$University of California, San Diego
\quad
$^{2}$University of Michigan, Ann Arbor
\quad
$^{3}$Boston College
\\[0.2em]
$^{4}$ShanghaiTech University
\quad
$^{5}$South China University of Technology
\quad
$^{6}$Tsinghua University
\\[0.2em]
$^{7}$Huazhong University of Science and Technology
\\[0.45em]
{\small
$^{*}$Co-first authors
\qquad
$^{\dagger}$Equal contribution
}
\end{tabular}
}

\begin{document}

\maketitle

\begin{abstract}
Collections of diverse observations often share one acquisition, processing, geometric, or channel corruption, while only an unpaired clean reference set is available. For a prescribed low-dimensional correction shared across observations, the observed and clean reference sets support inference only through the response of fixed moments. We formulate this problem as two-sample moment calibration and report a rank-aware information state combining local rank, scaled moment sensitivity, source-separated covariance, and a moment compatibility residual. Full rank gives local moment identifiability, whereas kernel directions remain unresolved to first order. A unified linearization separates observed-set and reference-set uncertainty. Under covariance weighting, the weakest scaled singular value determines worst-direction asymptotic amplification. For an orientation-preserving planar-similarity correction shared across observations, ensemble centroids and a nonzero third-order complex moment yield closed-form global population identification of translation, rotation, and isotropic scale under matched-population and no-clipping assumptions. Controlled validation tests the predicted $N^{-1}$ and $\sigma_{\min}^{-2}$ laws, Gaussian efficiency, and interval coverage. Bounded applications report color corrected-output quality, channel magnitude-response calibration, and a separate paired geometric de-beautification result. The framework therefore reports missing or weak information instead of treating every fitted correction as identified.
\end{abstract}

\section{Introduction}

Many AI pipelines process diverse observations affected by one shared acquisition, processing, or channel corruption. Paired clean targets may be unavailable, while an independent sample from the matched clean population is accessible. Repeated action of the shared corruption enables set-level calibration. Yet additional samples sharpen only moment responses that expose a correction direction. We therefore ask when an observed set and an unpaired clean reference set support a prescribed low-dimensional correction and which directions remain unresolved.


Structured blind inverse methods recover content and operators under subspace, sparsity, architectural, or learned priors \cite{ahmed2014blind,li2016identifiability,ulyanov2018dip,chung2023dps,chung2023blinddps,li2025blindpnp}. Unpaired methods learn transformations or degradations through supervision at the distribution level \cite{zhu2017cyclegan,shaham2017mmd,meanti2025ddm}. Moment estimation and local identification theory analyze empirical restrictions and derivative rank \cite{hansen1982gmm,rothenberg1971identification,vandervaart1998asymptotic}. The gap is the absence of an operational report combining the moment compatibility residual, local rank, scaled moment sensitivity, and two-sample precision for this two-population calibration problem.

Our thesis is that an observed set and an unpaired clean reference set support a prescribed low-dimensional correction only through the response of fixed moments. We make three contributions. First, we combine the corrected-moment Jacobian, scaled moment sensitivity, source-separated covariance, and moment compatibility residual in a rank-aware information state (Figure~\ref{fig:framework}). Controlled validation tests its rank, conditioning, and two-sample precision laws. Second, an analytic specialization gives closed-form global population identification for an orientation-preserving planar-similarity correction shared across observations under matched-population, no-clipping, and nonzero-moment assumptions. Third, bounded applications report color corrected-output quality, channel magnitude-response calibration, and paired geometric de-beautification. The paired application is not used to validate the unpaired planar-similarity result.

\begin{figure*}[t]
\centering
\includegraphics[width=0.98\textwidth]{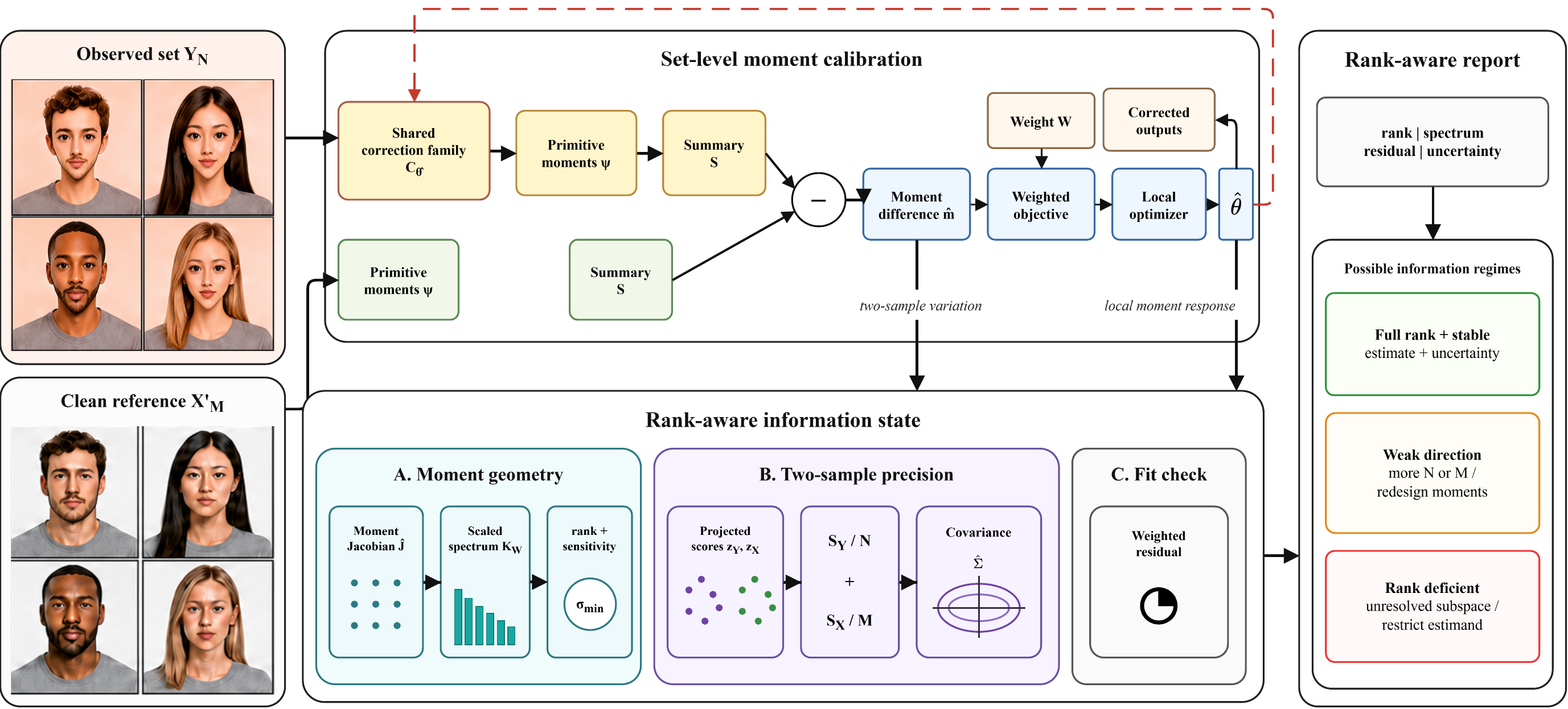}
\caption{Shared-corruption calibration from unpaired sets. One correction acts across diverse observations, and an independent clean reference set supplies reference moments. Moment matching estimates the correction. Moment response and two-sample variation together define the rank-aware information state.}
\label{fig:framework}
\end{figure*}

\section{Related Work}

\paragraph{Shared operators and unpaired calibration.}
Classical blind deconvolution derives identifiability from subspace or sparsity structure \cite{ahmed2014blind,li2016identifiability}. Deep image and diffusion priors couple content recovery with architectural or learned image and operator models \cite{ulyanov2018dip,chung2023dps,chung2023blinddps,li2025blindpnp}. Cycle consistency and distribution matching learn transformations from unpaired distributions \cite{zhu2017cyclegan,shaham2017mmd}. Diffusion Distribution Matching estimates a fixed degradation from small unpaired sets under its model assumptions \cite{meanti2025ddm}. Second-order channel methods identify linear responses from output statistics under excitation assumptions \cite{tong1994blind}. Our setting instead fixes a prescribed low-dimensional correction and audits information in selected moments against an independent reference population.

\paragraph{Moment estimation and local information.}
Generalized method of moments and extremum estimation establish asymptotic analysis for empirical restrictions and random objectives \cite{hansen1982gmm,pakes1989simulation,newey1994large,vandervaart1998asymptotic}. Local identification theory connects derivative rank with parameter uniqueness in a neighborhood \cite{rothenberg1971identification}. We apply these results to a corrected-moment map whose target is estimated from an independent clean reference set. Its Jacobian separates locally visible, weak, and first-order unresolved directions, while projected covariance separates the observed-set and reference-set contributions to two-sample precision. These established quantities serve an integrated reporting role rather than constituting new GMM, rank, or covariance results.


\paragraph{Geometric moments and registration.}
Classical moment invariants and complex-moment normalization characterize planar translation, rotation, and scale for one image, while Fourier registration estimates similarity transforms between corresponding images \cite{hu1962visual,abumostafa1985image,reddy1996fft}. Recent streaming work uses centroids and complex moments to identify congruence between finite point multisets and avoid vanishing rotation moments \cite{chang2026streaming}. We borrow these equivariance identities for a different statistical object. One transformation is shared across heterogeneous observations, the clean reference set is an independent population sample, and both sets contribute sampling uncertainty. The distinction lies not in a new complex-moment identity but in its bounded use within a two-population information audit for a prescribed shared correction.

\section{Problem Setup and Estimator}

The framework fixes one correction family and one set of moments. It then separates correction specification, information support, and two-sample precision. One parameter is shared by all observed samples. The estimator fits that parameter. The information state reports whether the prescribed moments support it and attributes uncertainty to the two finite sets. This separation lets photometric, geometric, and channel corrections share one estimator without treating their estimands as interchangeable.

\paragraph{Calibration signal and estimand.}
Let $Y_1,\ldots,Y_N\stackrel{\mathrm{iid}}{\sim}P_Y$ be observations affected by one shared corruption. Independently, let $X'_1,\ldots,X'_M\stackrel{\mathrm{iid}}{\sim}P_X$ be an unpaired clean reference set. A known differentiable family $C_\theta$, indexed by $\theta\in\Theta\subset\R^p$, acts on observations as a correction. Choose primitive moments $\psi:\mathcal X\to\R^\ell$ and a differentiable summary map $\mathcal S:\R^\ell\to\R^q$. Define $u_Y(\theta)=\E[\psi(C_\theta(Y))]$, $u_X=\E[\psi(X)]$, and $g(\theta)=\mathcal S\{u_Y(\theta)\}$.
The population restriction is
\begin{equation}
\mathcal S\{u_Y(\theta^*)\}=\mathcal S(u_X).
\label{eq:model}
\end{equation}
This restriction does not require sample-level correspondence or equality of the full corrected and clean distributions.

The identity summary gives ordinary moment matching, while linear projections, power-spectrum bins, and coordinate-sensitive spatial moments define correction-specific moment maps. For each scalar feature $f$, mean and standard-deviation matching uses the primitive pair $(f,f^2)$ and $\mathcal S(\mu,\nu)=(\mu,\sqrt{\nu-\mu^2})$ on positive-variance coordinates. The moments and summary map are frozen before final estimation and thereby define the information being audited.

\paragraph{Coordinate-sensitive moments for geometry.}
Spatially pooled appearance moments cannot identify geometric coordinates that they remove. For a geometric correction, let $\rho_x$ be a fixed normalized nonnegative spatial measure associated with image $x$ on $\mathbb C\simeq\R^2$, such as normalized luminance or feature energy. The construction of $\rho_x$ is fixed before calibration. Define the centroid and a third-order complex central moment by
\begin{equation}
\begin{aligned}
\mu(x)&=\int z\,d\rho_x(z),\\
\kappa(x)&=\int (z-\mu(x))^2
\overline{(z-\mu(x))}\,d\rho_x(z).
\end{aligned}
\label{eq:spatial_moments}
\end{equation}
For $a\in\mathbb C\setminus\{0\}$ and $b\in\mathbb C$, let $C_{a,b}$ denote the orientation-preserving planar-similarity correction; its induced spatial action is $T_{a,b}(z)=az+b$. On an unbounded or sufficiently padded domain with no clipping, $\rho_{C_{a,b}x}=(T_{a,b})_\#\rho_x$, and
\begin{equation}
\mu(C_{a,b}x)=a\mu(x)+b,
\qquad
\kappa(C_{a,b}x)=|a|^2a\,\kappa(x).
\label{eq:similarity_equivariance}
\end{equation}
Thus translation, rotation, and isotropic scale are represented by equivariant rather than invariant moments. Using the real and imaginary parts of $(\mu,\kappa)$ gives four real moment coordinates for the four real parameters in $(a,b)$.

\paragraph{Two-sample estimator.}
The empirical moment difference retains fluctuations from both sets:
\begin{equation}
\begin{aligned}
\widehat m_{N,M}(\theta)
&=\mathcal S\{\bar u_Y(\theta)\}-\mathcal S(\bar u_X),\\
\bar u_Y(\theta)&=\frac1N\sum_i\psi(C_\theta(Y_i)),\\
\bar u_X&=\frac1M\sum_j\psi(X'_j).
\end{aligned}
\label{eq:moment}
\end{equation}
Its population counterpart vanishes at $\theta^*$, and the smooth summary preserves the $N^{-1/2}$ and $M^{-1/2}$ first-order fluctuations. For a fixed symmetric positive-definite $W\in\R^{q\times q}$, define
\begin{equation}
\widehat{\mathcal L}_{N,M}(\theta)
=\widehat m_{N,M}(\theta)^\top W\widehat m_{N,M}(\theta).
\label{eq:estimator}
\end{equation}
We let $\widehat\theta$ denote the selected interior local minimizer; the theory below conditions on its consistency. The operator instances use either direct moment inversion or constrained gradient optimization.

\paragraph{Projected two-sample uncertainty.}
At the selected local minimizer, define
\begin{equation}
\widehat J
=\left.\nabla_\theta\mathcal S\{\bar u_Y(\theta)\}\right|_{\widehat\theta},
\qquad
\widehat A=\widehat J^\top W\widehat J.
\label{eq:plugin}
\end{equation}
Let $\widehat R_Y=D_u\mathcal S\{\bar u_Y(\widehat\theta)\}$ and $\widehat R_X=D_u\mathcal S(\bar u_X)$. The projected observed terms are $\widehat J^\top W\widehat R_Y\psi(C_{\widehat\theta}(Y_i))$; the reference terms are $\widehat J^\top W\widehat R_X\psi(X'_j)$. Let $\widehat S_Y$ and $\widehat S_X$ denote their sample covariances. When $\widehat J$ has full column rank, the projected sandwich estimator is
\begin{equation}
\widehat\Sigma_\theta(N,M)=
\widehat A^{-1}
\left(\frac{\widehat S_Y}{N}+\frac{\widehat S_X}{M}\right)
\widehat A^{-1}.
\label{eq:plugincov}
\end{equation}
After projection, covariance estimation and inversion involve only $p\times p$ matrices. The expanded score formulas and finite-sample standard-deviation convention are given in the Supplementary Material.

\paragraph{Rank-aware information state.}
Fix a diagonal matrix $D$ before estimation. Its entries $d_k=D_{kk}$ satisfy $0<d_k<\infty$ and encode declared parameter tolerances or reciprocal coordinate response norms computed on a pilot split. Let $J^*=\nabla_\theta g(\theta^*)\in\R^{q\times p}$ be the population Jacobian. The scaled weighted population and empirical Jacobians are
\begin{equation}
\begin{aligned}
K_W&=W^{1/2}J^*D,\\
\widehat K_W&=W^{1/2}\widehat J D.
\end{aligned}
\label{eq:whitenedJ}
\end{equation}
A full-rank fit reports $\widehat\theta$, the numerical rank and complete singular spectrum of $\widehat K_W$, the weighted residual, and $\widehat\Sigma_\theta(N,M)$. A rank-deficient fit reports the unresolved right singular subspace instead of ordinary parameter intervals.

The necessary dimension condition $q\ge p$ is not sufficient. The summaries must respond independently across parameter coordinates. A pilot Jacobian checks these responses on representative observations, and any data-selected moments are frozen using a pilot split or valid cross-fitting. Rank, spectrum, residual, and covariance define the operational distinctions among moment compatibility, scaled moment sensitivity, and two-sample precision; the retained operator instances do not report them jointly.

\paragraph{Standing assumptions and inference boundary.}
The restriction \eqref{eq:model} defines the estimand through the fixed correction family, moments, and summary map. The general theory assumes this restriction is correctly specified at an interior $\theta^*$; $g$ and $\mathcal S$ are continuously differentiable in the required neighborhoods; $\psi(C_\theta(Y))$ is almost surely differentiable with derivative dominated by an integrable envelope; primitive moments have finite second moments; the two sets are independent and i.i.d.\ within sets; empirical primitive moments and derivatives converge locally uniformly; and the selected interior local minimizer is consistent as $N,M\to\infty$. Full column rank is imposed only where ordinary local parameter inference is claimed. Corollary~\ref{cor:scaled} additionally requires a positive-definite leading moment covariance and fixes $W$ to its inverse. These are identification and inference assumptions, not guarantees of numerical convergence or correct specification in an application.

\section{Moment Geometry for Identifiability, Stability, and Two-Sample Precision}

The corrected-moment Jacobian describes which local parameter directions are visible and how sampling noise is amplified. The two-sample influence covariance describes how observed and reference variation passes through that geometry into parameter precision.
Proof ideas are retained below; complete proofs and the optional joint Gaussian limit are given in the Supplementary Material.

\begin{proposition}[Local moment identifiability]
\label{prop:local}
If $J^*$ has full column rank, then there is a neighborhood $U$ of $\theta^*$ in which $g(\theta)=g(\theta^*)$ implies $\theta=\theta^*$.
\end{proposition}
\noindent\textit{Proof idea.}
Select $p$ rows of $J^*$ forming a nonsingular matrix and apply the inverse-function theorem to the corresponding coordinates of $g$.

For any $v\in\ker J^*$, differentiability gives
\begin{equation}
\|g(\theta^*+tv)-g(\theta^*)\|_2=o(|t|),
\qquad t\to0.
\label{eq:kernel}
\end{equation}
Thus full rank supports local parameter inference, while a Jacobian-kernel direction is unresolved to first order. Increasing $N$ or $M$ reduces sampling noise but does not change this population derivative.

\begin{theorem}[Unified two-sample linearization]
\label{thm:asymptotic}
Suppose the standing conditions hold and $J^*$ has full column rank. Write $u_Y^*=u_Y(\theta^*)$, $R_Y^*=D_u\mathcal S(u_Y^*)$, and $R_X=D_u\mathcal S(u_X)$. Define
\begin{equation}
\begin{aligned}
a_Y(Y)&=R_Y^*\{\psi(C_{\theta^*}(Y))-u_Y^*\},\\
a_X(X')&=R_X\{\psi(X')-u_X\}.
\end{aligned}
\end{equation}
Let $V_Y=\Var[a_Y(Y)]$, $V_X=\Var[a_X(X')]$, and
\begin{equation}
B=(J^{*\top}WJ^*)^{-1}J^{*\top}W.
\end{equation}
Then the selected consistent local minimizer satisfies
\begin{equation}
\begin{split}
\widehat\theta-\theta^*={}&-B\left[
\frac1N\sum_{i=1}^Na_Y(Y_i)\right.\\
&\left.-\frac1M\sum_{j=1}^Ma_X(X'_j)
\right]+o_p(N^{-1/2}+M^{-1/2}).
\end{split}
\label{eq:linearization}
\end{equation}
Its leading asymptotic covariance is
\begin{equation}
\Sigma_{N,M}
=B\left(\frac{V_Y}{N}+\frac{V_X}{M}\right)B^\top.
\label{eq:sandwich}
\end{equation}
\end{theorem}
\noindent\textit{Proof idea.}
The delta method gives the two independent influence sums; linearizing the GMM first-order condition maps their additive covariance through $B$ \cite{hansen1982gmm,vandervaart1998asymptotic}.

Theorem~\ref{thm:asymptotic} maps moment noise into parameter error through $B$. The covariance separates the observed contribution $BV_YB^\top/N$ from the reference contribution $BV_XB^\top/M$. The same Jacobian that certifies local identifiability therefore controls how both sampling sources are amplified.

\begin{corollary}[Scaled asymptotic error amplification]
\label{cor:scaled}
For a planned pair $(N,M)$, let $\Omega_{N,M}=V_Y+(N/M)V_X$ be positive definite and fix $W=\Omega_{N,M}^{-1}$. Let $e_D=D^{-1}(\widehat\theta-\theta^*)$. The leading covariance $\Sigma_D$ of $e_D$ satisfies
\begin{equation}
\begin{aligned}
\Sigma_D&=\frac1N(K_W^\top K_W)^{-1},\\
\sup_{\|u\|_2=1}u^\top\Sigma_Du
&=\frac{1}{N\sigma_{\min}(K_W)^2}.
\end{aligned}
\label{eq:scaled_amplification}
\end{equation}
\end{corollary}
\noindent\textit{Proof idea.}
Substitute $W=\Omega_{N,M}^{-1}$ into Equation~\eqref{eq:sandwich}, conjugate by $D^{-1}$, and apply the largest-eigenvalue variational characterization.

\begin{proposition}[Population identification of a shared planar similarity]
\label{prop:similarity}
Assume that applying the correction $C_{a^*,b^*}$ to the observed population yields the clean population,
\[
C_{a^*,b^*}(Y)\stackrel{d}{=}X,
\qquad
a^*\in\mathbb C\setminus\{0\},
\quad
b^*\in\mathbb C.
\]
Let superscripts $Y$ and $X$ denote observed- and clean-population moments:
\[
\begin{aligned}
m_\mu^Y&=\E[\mu(Y)], & m_\mu^X&=\E[\mu(X)],\\
m_\kappa^Y&=\E[\kappa(Y)], & m_\kappa^X&=\E[\kappa(X)].
\end{aligned}
\]
If $m_\kappa^Y\neq0$, then $(a^*,b^*)$ is globally identified at the population level within the orientation-preserving similarity family by
\begin{equation}
r^*=\frac{m_\kappa^X}{m_\kappa^Y},
\qquad
a^*=\frac{r^*}{|r^*|^{2/3}},
\qquad
b^*=m_\mu^X-a^*m_\mu^Y.
\label{eq:similarity_inverse}
\end{equation}
The corresponding four-real-parameter moment map has full column rank. Writing hats for the corresponding observed-set and reference-set averages, the plug-in estimator exists whenever $\widehat m_\kappa^Y\neq0$. Under finite second moments, it is consistent and inherits the separate $N^{-1}$ and $M^{-1}$ first-order covariance contributions in Theorem~\ref{thm:asymptotic}.
\end{proposition}
\noindent\textit{Proof idea.}
Equivariance gives the two moment equations. The map $a\mapsto|a|^2a$ is bijective off zero, while $m_\kappa^Y\neq0$ and the centroid equation make the full real Jacobian nonsingular.

The condition $m_\kappa^Y\neq0$ is substantive. If this ensemble moment vanishes because of rotational symmetry or cross-image cancellation, the prescribed moment pair does not identify orientation and scale. More samples can reduce estimation noise around zero but cannot create the missing response. The rank-aware action is to add another equivariant moment or restrict the estimand. Proposition~\ref{prop:similarity} covers shared planar translation, rotation, and isotropic scale on a no-clipping domain. General affine and nonrigid families retain only the local claim in Proposition~\ref{prop:local} unless their chosen moment map is shown to be globally injective.

\begin{corollary}[Gaussian mean shift]
\label{cor:gaussian}
Let $\psi(x)=x$, let $\mathcal S$ be the identity map, and let $C_\theta(y)=y-\theta$. Suppose
\begin{equation}
Y_i\stackrel{\mathrm{iid}}{\sim}\mathcal N(\mu+s,VI_d),
\qquad
X'_j\stackrel{\mathrm{iid}}{\sim}\mathcal N(\mu,VI_d),
\end{equation}
with independent sets, $V>0$, and $S=\|s\|_2>0$. Then $\widehat s=\overline Y-\overline X'$ is unbiased and
\begin{equation}
\E\|\widehat s-s\|_2^2
=dV\left(\frac1N+\frac1M\right).
\end{equation}
A target relative MSE $\E\|\widehat s-s\|_2^2/S^2\le\tau^2$ is feasible when $dV/(MS^2)<\tau^2$. The observed sample requirement is
\begin{equation}
N\ge
\left\lceil
\frac{dV}{\tau^2S^2-dV/M}
\right\rceil.
\label{eq:sample_requirement}
\end{equation}
When $\mu$ is known, set $M=\infty$. The estimator $\overline Y-\mu$ then attains the Cram\'er--Rao lower bound $VI_d/N$.
\end{corollary}
\noindent\textit{Proof idea.}
The difference of the independent Gaussian means has covariance $V(1/N+1/M)I_d$; its trace gives the MSE, and the known-mean Fisher information is $NI_d/V$.

The theory yields five stochastic predictions for the retained experiments. Parameter error decreases with $N$ only within locally sensed directions. Larger individual variance and weaker shared effects require more samples. A finite reference set contributes a separate covariance term. Smaller scaled singular values amplify parameter error according to $\sigma_{\min}^{-2}$. In the exact linear construction, Jacobian-kernel components remain unchanged under the standardized minimum-norm estimator. Proposition~\ref{prop:similarity} gives a separate analytic result. Under its matched-population, no-clipping, and nonzero-moment assumptions, one shared orientation-preserving planar similarity is globally identified at the population level.

\begin{figure*}[t]
\centering
\includegraphics[width=0.98\textwidth]{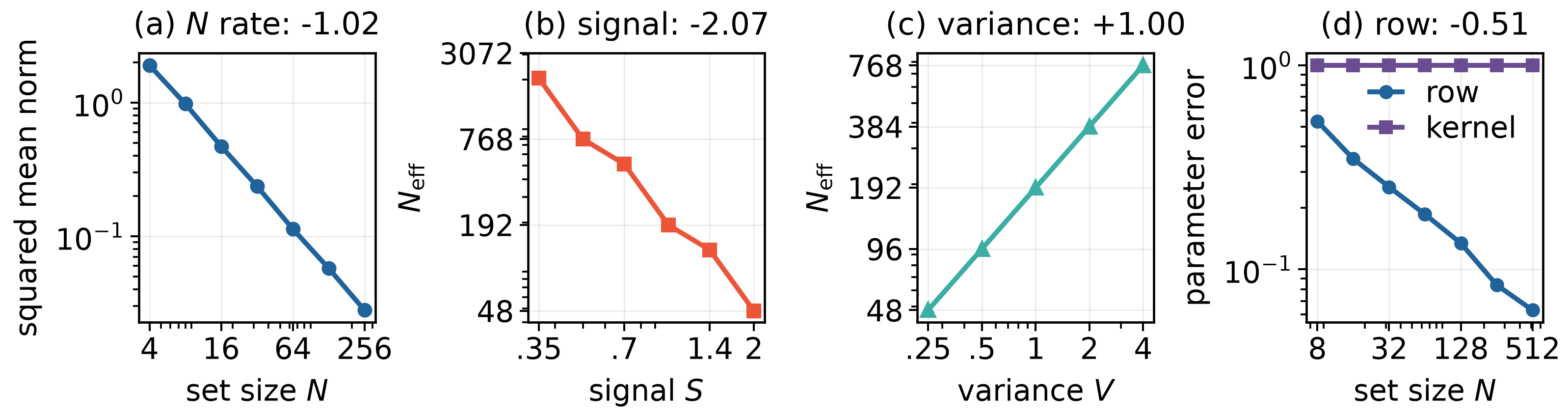}
\caption{Sampling improves sensed directions but leaves exact Jacobian-kernel directions unresolved. (a) Squared mean Euclidean shift-estimation error decreases with set size. (b) Required set size decreases with signal magnitude. (c) Required set size increases with individual variance. (d) Row-space error decreases while the injected kernel component remains unchanged. Slopes are log-log fits.}
\label{fig:theory}
\end{figure*}

\section{Experiments}

The controlled constructions test moment visibility, scaled moment sensitivity, and two-sample precision. The color and channel protocols report bounded operator outcomes, while Figure~\ref{fig:debeautification} provides a separate paired, image-level geometric de-beautification result. None of these application protocols validates the complete rank-aware information state or Proposition~\ref{prop:similarity}.

\paragraph{Controlled moment visibility and precision.}
The rank-deficient construction tests whether sampling reduces error only in sensed directions. It drew $x\sim\mathcal N(0,I_6)$ and used $y=x+G\theta^*$ and $C_\theta(y)=y-G\theta$, with $G\in\R^{6\times p}$. For $Q\in\R^{k\times6}$, the audited statistic was $Q\bar c_\theta$, where $\bar c_\theta=N^{-1}\sum_i C_\theta(y_i)$, so its Jacobian was $J=-QG$. Parameter coordinates were standardized before the minimum-norm fit. The test retained $k=2$ moment directions for $p=4$ parameters over 400 repetitions. In Figure~\ref{fig:theory}(d), mean row-space error fell from $0.5314$ to $0.0631$ with slope $-0.509$ and $R^2=0.997$, while the injected unit kernel component remained $1.000$. A separate $k=5$ draw had rank four and reduced total error from $0.6878$ to $0.0901$. Thus finite noise decreases along sensed directions, whereas a component absent from the selected moment map remains unresolved.

The scaling construction tests the predicted set-size, signal-strength, and variance exponents. It used $d=8$, $M=40{,}000$, and 200 repetitions. Squared mean Euclidean error decreased from $1.9011$ to $0.0280$ across the tested observed-set sizes. For Figures~\ref{fig:theory}(b)--(c), $N_{\mathrm{eff}}$ was the first tested set size with mean relative error below $0.20$. The fitted slopes were $-1.017$, $-2.068$, and $1.000$, respectively. These regressions test the exact linear/Gaussian construction rather than asserting the same finite-sample exponents for the nonlinear operator instances.

\begin{figure*}[t]
\centering
\includegraphics[width=0.98\textwidth]{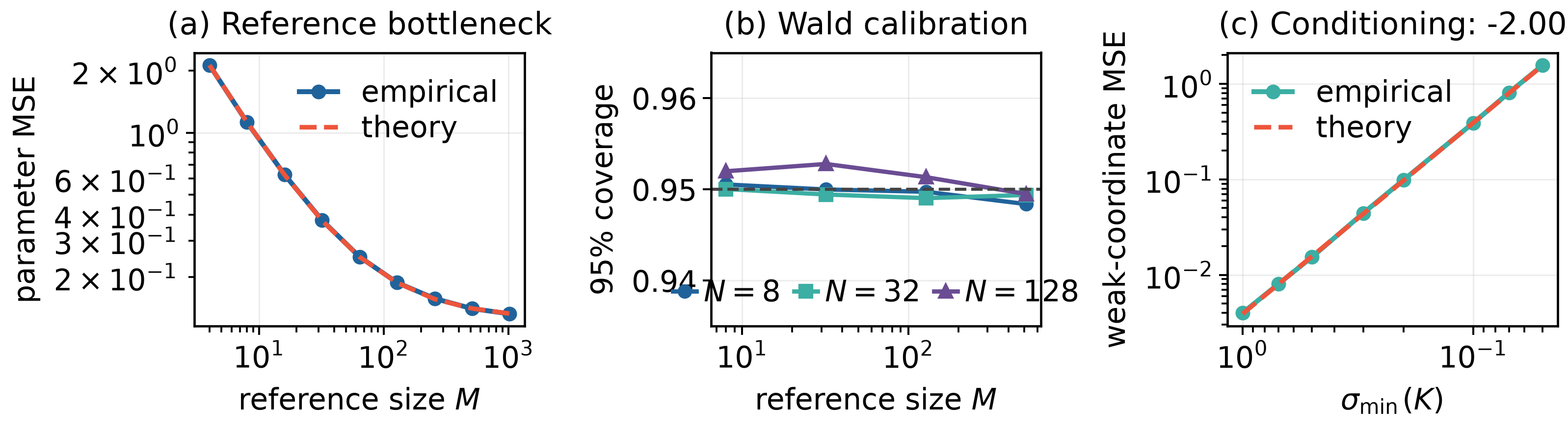}
\caption{Reference size and local conditioning control distinct components of two-sample precision. Each setting uses 20,000 draws. (a) With $N=64$, empirical mean-shift MSE follows $dV(1/N+1/M)$ across reference sizes. (b) Exact Gaussian 95\% intervals retain nominal scalar coverage across observed and reference sizes. (c) A one-dimensional optimally weighted linear test follows $1/\{N\sigma_{\min}(K_W)^2\}$.}
\label{fig:uncertainty}
\end{figure*}

The reference-size study used $d=8$, $V=1$, $N=64$, and 20,000 exact Gaussian sample-mean draws per reference size. Figure~\ref{fig:uncertainty}(a) matched $dV(1/N+1/M)$ within $0.58\%$ and approached the observed-set floor $dV/N$ as $M$ increased. The scalar intervals were $\widehat s_1\pm1.96\sqrt{V(1/N+1/M)}$; their empirical coverage ranged from $94.84\%$ to $95.28\%$ across the tested $N$ and $M$. This coverage statement is confined to the exact Gaussian construction. A separate known-mean study with $d=16$, $V=1$, and 4,000 repetitions gave CRB ratios from $0.986$ to $1.008$. In the full-rank weak-direction test, $N=256$ and $\sigma_{\min}(K_W)$ ranged from $1$ to $0.05$; MSE followed \eqref{eq:scaled_amplification} with slope $-1.999$ and $R^2=1.000$. Together, the two panels distinguish an insufficient reference set from an ill-conditioned but full-rank moment response. Complete grids and seeds are in the Supplementary Material.

\paragraph{Color corrected-output calibration.}
The color instance resized images to $256\times256$ and used the fixed correction family
\[
C_\theta=G_\eta\circ S_s\circ\Gamma_\gamma\circ R_\beta\circ A_\alpha.
\]
The same stages acted on every image and clamped outputs to $[0,1]$. The affine, monotone piecewise-linear, gamma, saturation, and gain--bias stages contributed $12$, $48$, $3$, $1$, and $6$ parameters, respectively, giving $p=70$. The piecewise-linear stage used 16 positive increments per channel. The fit matched pooled means and standard deviations from RGB and VGG-16 features against summaries from 200 clean reference faces. RGB contributed six moment coordinates. The 256-channel \texttt{features[:16]} output contributed 512, giving $q=518$. VGG-16 used fixed ImageNet-1K V1 weights, standard normalization, and evaluation mode. RGB summaries were multiplied by five and the objective used $W=I$. The correction was initialized at the identity and fitted for 300 Adam steps at learning rate $0.05$. These choices define the reported operator instance; full pooling and optimization conventions are in the Supplementary Material.

The color protocol kept the training, test, and reference identities disjoint. Images were $256\times256$ unretouched FFHQ faces \cite{karras2019stylegan}. Nine nonlearned Instagram-style formulas from \texttt{pilgram} were applied to obtain paired filtered images. Positions 0--1999 in the sorted identity list formed the training pool, while positions 2000--2299 formed the identity-disjoint test pool. Learned baselines used 16 filters on 1800 training and 200 validation identities. The nine test filters---aden, brooklyn, inkwell, lark, maven, moon, rise, slumber, and stinson---were disjoint from those 16. Importantly, all 25 belong to the same \texttt{pilgram} implementation family rather than different apps or engines. The clean reference set comprised 200 further FFHQ faces, at sorted indices 15000--15199, disjoint from both pools.

For each test filter, the first 30 filtered images formed the observed set; their clean counterparts were never used for fitting. The next 40 rows were queries, with clean images used only as evaluation ground truth, giving 360 queries. The set-size sweep applied the same rule at each $N$: fit on the first $N$ filtered images and evaluate on the next 40.

\begin{figure*}[t]
\centering
\includegraphics[width=0.98\textwidth]{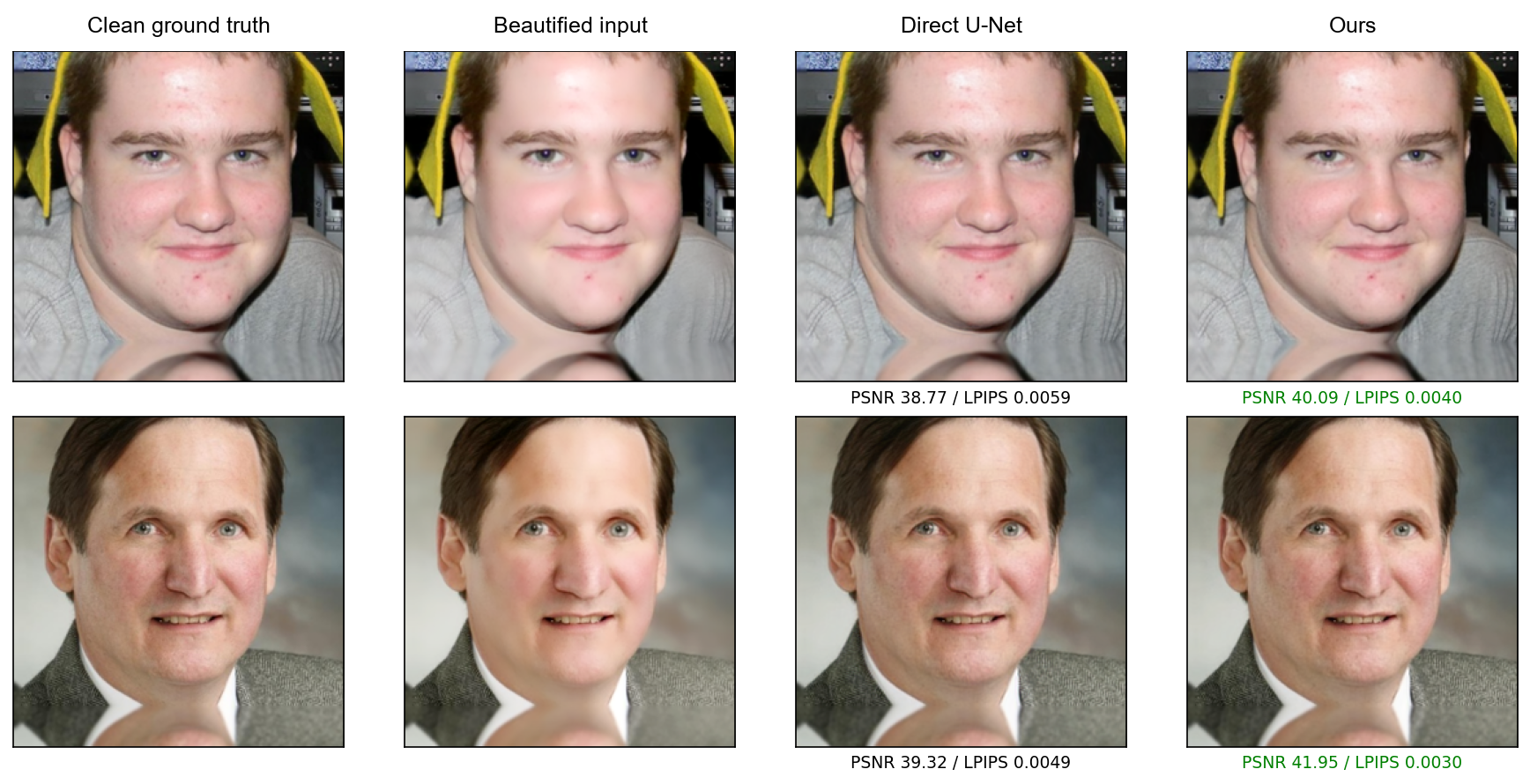}
\caption{Paired geometric de-beautification on two test faces. Columns show clean ground truth, the commercial-API beautified input, a task-specific Direct U-Net, and Ours. The printed per-image PSNR/LPIPS values favor Ours on both examples. The full four-example panel is in the Supplementary Material.}
\label{fig:debeautification}
\end{figure*}

Across shared-filter sets, LPIPS decreased from $0.1788$ at $N=1$ to $0.0578$ at $N=30$. Replacing the shared-filter calibration set with mixed filters increased the $N=30$ value to $0.1150$.

LPIPS was $0.0954$ for filtered input, $0.1023$ for a fixed supervised ARColorOpV2 trained on paired data from the 16 seen filters, $0.0761$ for set-level Reinhard transfer \cite{reinhard2001color}, and $0.0547$ for moment calibration. Harm frequency was $70\%$, $17\%$, and $1\%$ for ARColorOpV2, Reinhard, and moment calibration, respectively; harm means corrected LPIPS exceeded filtered-input LPIPS. These rows use different training and adaptation information and are not a universal restoration ranking.

\paragraph{Paired geometric de-beautification.}
For an image-level application, we evaluated our geometry-aware de-beautification model on approximately 480 paired $256\times256$ FFHQ faces processed by one commercial face-reshape operator. A width-48 U-Net predicts a stationary velocity field and an appearance residual; six scaling-and-squaring steps integrate the field before the warped input and residual are combined. Training used 5,000 Adam updates with batch size 16 and learning rate $2\times10^{-4}$. The objective combined LPIPS, pixel, identity, field-magnitude, folding, and residual penalties. The first 85\% of sorted pairs were used for training and the remaining 72 for testing. The width-matched Direct U-Net used the same pairs and split, directly predicted a residual image, and trained for 4,000 updates. On this in-distribution test, Ours reached PSNR $39.83$ and LPIPS $0.0049$, versus $38.34$ and $0.0061$ for Direct U-Net; SSIM tied at $0.972$, while identity similarity was $0.997$ versus $0.995$. Figure~\ref{fig:debeautification} shows representative outputs. This paired, per-image protocol demonstrates restored face geometry but is separate from the unpaired estimator and the shared-planar-similarity assumptions of Proposition~\ref{prop:similarity}.

\paragraph{Channel magnitude-response calibration.}
The channel instance passed length-256 AR(1) signals through one fixed normalized eight-tap LTI channel shared by every observed signal. For a stationary input,
\[
P_Y(\omega)=|H(\omega)|^2P_X(\omega).
\]
On evaluated frequencies, $P_X(\omega)\ge c_X>0$ and $|H(\omega)|\ge c_H>0$. The sampled magnitude response is the estimand, its reciprocal defines the correction, and the reference estimate used 4,000 clean signals.
The implemented estimator was
\begin{equation}
\widehat{|H(\omega)|}
=\left[
\max\left\{
\frac{\widehat P_Y(\omega)}
{\widehat P_X(\omega)+10^{-9}},\,10^{-8}
\right\}
\right]^{1/2}.
\label{eq:channel_estimator}
\end{equation}
Same-mode convolution and AR initialization introduce finite-length boundary effects. The finite records therefore do not exactly realize the stationary identity. Phase remains unresolved, and no reconstructed-signal metric is claimed. Averaged over 30 trials, relative magnitude-response MSE decreased from $0.2340$ at $N=1$ to $0.0012$ at $N=256$. For $N\ge2$, the log-log slope was $-0.971$ with $R^2=1.000$; the full grid is in the Supplementary Material.

\section{From Moment Geometry to Calibration Decisions}

Before final fitting, freeze the correction family, moments $(\psi,\mathcal S)$, parameter scales $D$, and weight $W$. Data-selected moments require a pilot split or valid cross-fitting. On the remaining samples, fit $\widehat\theta$, compute the full spectrum of $\widehat K_W$, record the numerical-rank tolerance and weighted residual, and use \eqref{eq:plugincov} only when the fitted Jacobian has full column rank. Finite differences along weak right singular vectors check the corresponding local moment changes. The report then maps each information state to an action. An adequately conditioned full-rank state returns local estimates and projected intervals. A full-rank state with small $\sigma_{\min}(\widehat K_W)$ assigns a high precision cost to weak directions. A rank-deficient state returns the fitted correction, residual, scaled spectrum, and unresolved right singular subspace rather than ordinary intervals. For Proposition~\ref{prop:similarity}, a zero or near-zero ensemble third-order moment warns that the selected moments do not support orientation and scale. Increasing $N$ or $M$ cannot remove that population symmetry. One must instead add an equivariant moment or restrict the estimand. For a prespecified contrast $u$ and target interval half width $r$, evaluate
\[
1.96\{u^\top\widehat\Sigma_\theta(N,M)u\}^{1/2}\le r
\]
over feasible sample allocations; the separate observed and reference terms identify which set offers the larger precision gain. The residual remains a compatibility diagnostic: a large value flags tension among the correction family, selected moments, and reference population, whereas a small value does not establish correct specification. Material operating-point changes require recomputing the rank-aware information state.

\section{Scope and Conclusion}

An observed set and an unpaired clean reference set support a prescribed shared correction only through fixed-moment response. The rank-aware information state separates local visibility, scaled moment sensitivity, two-sample precision, and compatibility. Full rank yields local moment identifiability. The planar-similarity result strengthens this to global population identification only under matched-population, no-clipping, and nonzero-moment assumptions. Controlled validation tests the predicted mechanisms. Bounded applications report color corrected-output quality, channel magnitude-response calibration, and paired in-distribution geometric de-beautification. The paired per-image result does not extend the unpaired theory to nonrigid deformation. The framework reports uncertainty or unresolved directions rather than unsupported identification; clipped, arbitrary nonrigid, and phase-recovery settings remain outside its general claims. The paired result uses one operator and source; cross-operator generalization and arbitrary geometric inversion are not claimed.

\bibliography{setlevel2027}

\end{document}


\maketitle

\section{Claim--Evidence Closure}

Table~\ref{tab:sup-closure} maps every formal result to its complete argument, controlled check when claimed, and operative boundary, then summarizes the three bounded applications. The stochastic predictions are tested in constructions that satisfy their assumptions, while the planar-similarity result is closed analytically through an explicit inverse and nonsingular Jacobian. At a glance, moment calibration gives LPIPS $0.0547$ with $1\%$ harm, channel magnitude-response error decreases by $99.5\%$ from $N=1$ to $N=256$, and paired Ours improves Direct U-Net by $1.49$ dB PSNR while reducing LPIPS by $19.7\%$.

\begin{table*}[t]
\centering
\scriptsize
\setlength{\tabcolsep}{4pt}
\begin{tabular}{@{}>{\raggedright\arraybackslash}p{0.18\textwidth}>{\raggedright\arraybackslash}p{0.27\textwidth}>{\raggedright\arraybackslash}p{0.27\textwidth}>{\raggedright\arraybackslash}p{0.22\textwidth}@{}}
\toprule
Main-paper claim & Analytic closure in this supplement & Controlled or application evidence & Boundary retained in the claim \\
\midrule
Local moment identifiability & A full-column-rank submap is locally invertible; Jacobian-kernel directions have zero first-order moment response. & Row-space error decreases from $0.5314$ to $0.0631$, while the injected unit kernel component remains $1.000$. & Local, fixed moments and correction family; no global injectivity follows in general. \\
Two-sample precision & The influence expansion separates $BV_YB^\top/N$ from $BV_XB^\top/M$, with the corresponding joint central limit law. & Exact Gaussian MSE matches $dV(1/N+1/M)$ within $0.58\%$; scalar coverage is $94.84$--$95.28\%$. & Independent observed and reference sets; coverage statement is for the exact Gaussian construction. \\
Scaled amplification & Under covariance weighting, the worst standardized variance is exactly $1/\{N\sigma_{\min}(K_W)^2\}$. & The weak-direction experiment gives slope $-1.999$ with $R^2=1.000$. & Full rank at the audited operating point; small singular values imply high precision cost. \\
Shared planar similarity & Centroid and third-order complex moment equations give a closed-form inverse; the four-real-parameter Jacobian is nonsingular. & Analytic identification is complete without using the nonrigid image experiment as a surrogate validation. & Matched population, orientation-preserving similarity, no clipping, and nonzero ensemble third-order moment. \\
Gaussian mean shift & The exact risk, feasible sample requirement, and known-mean Cram\'er--Rao efficiency follow from the difference of independent sample means. & Empirical CRB ratios range from $0.986$ to $1.008$; set-size, signal, and variance slopes are $-1.017$, $-2.068$, and $1.000$. & Isotropic Gaussian location model for the exact finite-sample statements. \\
\midrule
\multicolumn{4}{@{}l}{\textit{Bounded application outcomes}} \\
\midrule
Color corrected output & Moment fit from a shared filtered set and an unpaired clean reference. & LPIPS $0.0547$, harm $1\%$; LPIPS is $42.7\%$ lower than the filtered input. & Nine held-out within-family filters; 360 queries. \\
Channel magnitude response & Periodogram-ratio estimator for the shared $|H(\omega)|$. & Relative MSE decreases from $0.2340$ to $0.0012$; slope $-0.971$. & Magnitude only; 30 trials per point. \\
Paired geometry & Velocity-field warping plus an appearance residual. & PSNR $39.83$, LPIPS $0.0049$ vs. Direct U-Net $38.34$, $0.0061$. & One operator; 72 in-distribution tests; paired and separate from the unpaired theory. \\
\bottomrule
\end{tabular}
\caption{Formal claim--evidence closure and bounded application summary. Each row preserves the assumptions and information regime under which its conclusion is supported.}
\label{tab:sup-closure}
\end{table*}

\section{Estimator Construction}

Notation and assumptions follow the main paper. At the selected local minimizer $\widehat\theta$, define
\[
\widehat R_Y=D_u\mathcal S\{\bar u_Y(\widehat\theta)\},
\qquad
\widehat R_X=D_u\mathcal S(\bar u_X).
\]
The estimated summary influences and their $p$-dimensional projections are
\begin{align*}
a_i^Y&=\widehat R_Y
\{\psi(C_{\widehat\theta}(Y_i))-\bar u_Y(\widehat\theta)\},
&
z_i^Y&=\widehat J^\top Wa_i^Y,\\
a_j^X&=\widehat R_X\{\psi(X'_j)-\bar u_X\},
&
z_j^X&=\widehat J^\top Wa_j^X.
\end{align*}
Let $\widehat S_Y$ and $\widehat S_X$ be the sample covariances of $\{z_i^Y\}_{i=1}^N$ and $\{z_j^X\}_{j=1}^M$, respectively. With $\widehat A=\widehat J^\top W\widehat J$, the projected sandwich estimator is
\[
\widehat\Sigma_\theta(N,M)=
\widehat A^{-1}
\left(\frac{\widehat S_Y}{N}+\frac{\widehat S_X}{M}\right)
\widehat A^{-1}.
\]
After projection, covariance estimation and inversion involve only $p\times p$ matrices.

For mean and standard-deviation matching, a pooled denominator-$L-1$ correction differs from the population-form standard deviation by a known $1+O(L^{-1})$ factor. It therefore preserves the first-order law used in the main paper.

\newpage
\section{Complete Proofs and Derivations}

\subsection{Local Moment Identifiability}

Full column rank of $J^*=\nabla_\theta g(\theta^*)$ gives a row-selection matrix $T\in\R^{p\times q}$ for which $TJ^*$ is nonsingular. The inverse-function theorem applied to $Tg(\theta)$ makes this selected map locally one-to-one at $\theta^*$. Equality of the full moment vectors implies equality of their selected coordinates and hence local equality of their parameters.

\subsection{Unified Two-Sample Linearization}

Write
\[
\widehat m_{N,M}(\theta)
=\mathcal S\{\bar u_Y(\theta)\}-\mathcal S(\bar u_X).
\]
Under the standing differentiability and local-uniform-convergence conditions, the delta method gives
\begin{equation}
\begin{aligned}
\widehat m_{N,M}(\theta^*)
&=\frac1N\sum_{i=1}^N a_Y(Y_i)\\
&\quad-\frac1M\sum_{j=1}^M a_X(X'_j)\\
&\quad+o_p(N^{-1/2}+M^{-1/2}).
\end{aligned}
\label{eq:sup-moment-expansion}
\end{equation}
The selected interior local minimizer satisfies the first-order condition
\[
0=\nabla_\theta\widehat m_{N,M}(\widehat\theta)^\top
W\widehat m_{N,M}(\widehat\theta).
\]
Expanding this condition at $\theta^*$ and using
$J^{*\top}WJ^*\succ0$ yields
\[
\begin{aligned}
\widehat\theta-\theta^*
&=-B\widehat m_{N,M}(\theta^*)
+o_p(N^{-1/2}+M^{-1/2}),\\
B&=(J^{*\top}WJ^*)^{-1}J^{*\top}W.
\end{aligned}
\]
Substitution of Equation~\eqref{eq:sup-moment-expansion} gives the linearization in Theorem~1 of the main paper. Independence of the two sets gives the leading covariance
\[
B\left(\frac{V_Y}{N}+\frac{V_X}{M}\right)B^\top.
\]
Finally, if $N/(N+M)\to\lambda\in[0,1]$, independent central limit theorems and Slutsky's theorem give
\[
\sqrt{\frac{NM}{N+M}}(\widehat\theta-\theta^*)
\rightsquigarrow
\mathcal N\!\left(
0,
B\{(1-\lambda)V_Y+\lambda V_X\}B^\top
\right).
\]

\subsection{Scaled Asymptotic Error Amplification}

For a planned pair $(N,M)$, define
$\Omega_{N,M}=V_Y+(N/M)V_X$ and fix
$W=\Omega_{N,M}^{-1}$. Then
\[
\frac{V_Y}{N}+\frac{V_X}{M}
=\frac{\Omega_{N,M}}{N},
\]
so the leading covariance of $\widehat\theta-\theta^*$ reduces to
\[
\Sigma_{N,M}=\frac1N(J^{*\top}WJ^*)^{-1}.
\]
Conjugating by $D^{-1}$ and using
$K_W=W^{1/2}J^*D$ gives
\[
\Sigma_D
=D^{-1}\Sigma_{N,M}D^{-1}
=\frac1N(K_W^\top K_W)^{-1}.
\]
The variational characterization of the largest eigenvalue then gives
\[
\sup_{\|u\|_2=1}u^\top\Sigma_Du
=\frac{1}{N\sigma_{\min}(K_W)^2}.
\]

\subsection{Population Identification of a Shared Planar Similarity}

The similarity equivariance identities give
\[
m_\mu^X=a^*m_\mu^Y+b^*,
\qquad
m_\kappa^X=|a^*|^2a^*m_\kappa^Y.
\]
Because $m_\kappa^Y\neq0$, define
$r^*=m_\kappa^X/m_\kappa^Y$. Writing
$a^*=s e^{\mathrm{i}\alpha}$ with $s>0$ gives
$r^*=s^3e^{\mathrm{i}\alpha}$ and therefore
\[
a^*=\frac{r^*}{|r^*|^{2/3}},
\qquad
b^*=m_\mu^X-a^*m_\mu^Y.
\]
The map $a\mapsto|a|^2a$ is a bijection on
$\mathbb C\setminus\{0\}$ and has real Jacobian determinant
$3|a|^4>0$. Multiplication by the nonzero
$m_\kappa^Y$ preserves nonsingularity, while the centroid equation contributes an identity block in $b$. The full four-real-parameter Jacobian is therefore nonsingular.

Replacing population moments by their observed-set or reference-set averages gives
\[
\widehat r=\frac{\widehat m_\kappa^X}{\widehat m_\kappa^Y},
\qquad
\widehat a=\frac{\widehat r}{|\widehat r|^{2/3}},
\qquad
\widehat b=\widehat m_\mu^X-\widehat a\,\widehat m_\mu^Y,
\]
whenever $\widehat m_\kappa^Y\neq0$. Consistency follows from the law of large numbers. The real-imaginary representation and the delta method give separate $N^{-1}$ and $M^{-1}$ first-order covariance contributions.

\subsection{Gaussian Mean Shift}

For independent Gaussian samples,
\[
\overline Y-\overline X'
\sim
\mathcal N\!\left(
s,\,
V\left(\frac1N+\frac1M\right)I_d
\right).
\]
Thus $\widehat s=\overline Y-\overline X'$ is unbiased and
\[
\E\|\widehat s-s\|_2^2
=dV\left(\frac1N+\frac1M\right).
\]
Feasibility requires $dV/(MS^2)<\tau^2$. For such targets, solving the relative-MSE inequality for $N$ gives
\[
N\ge
\left\lceil
\frac{dV}{\tau^2S^2-dV/M}
\right\rceil.
\]
When $\mu$ is known, set $M=\infty$. The Fisher information for $s$ is $NI_d/V$, whose inverse equals the covariance $VI_d/N$ of $\overline Y-\mu$. This estimator therefore attains the Cram\'er--Rao lower bound in the stated known-mean Gaussian model.

\section{Additional Experimental Details}

\subsection{Controlled Validation}

All synthetic studies used fixed random seeds. The scaling, Cram\'er--Rao, and row-kernel studies used 200, 4,000, and 400 repetitions, respectively. Each two-sample precision setting used 20,000 independent Monte Carlo repetitions. Figures~\ref{fig:sup-theory} and~\ref{fig:sup-uncertainty} collect the full visual evidence for the predicted set-size, visibility, reference-noise, and conditioning mechanisms.

\begin{figure*}[t]
\centering
\includegraphics[width=0.98\textwidth]{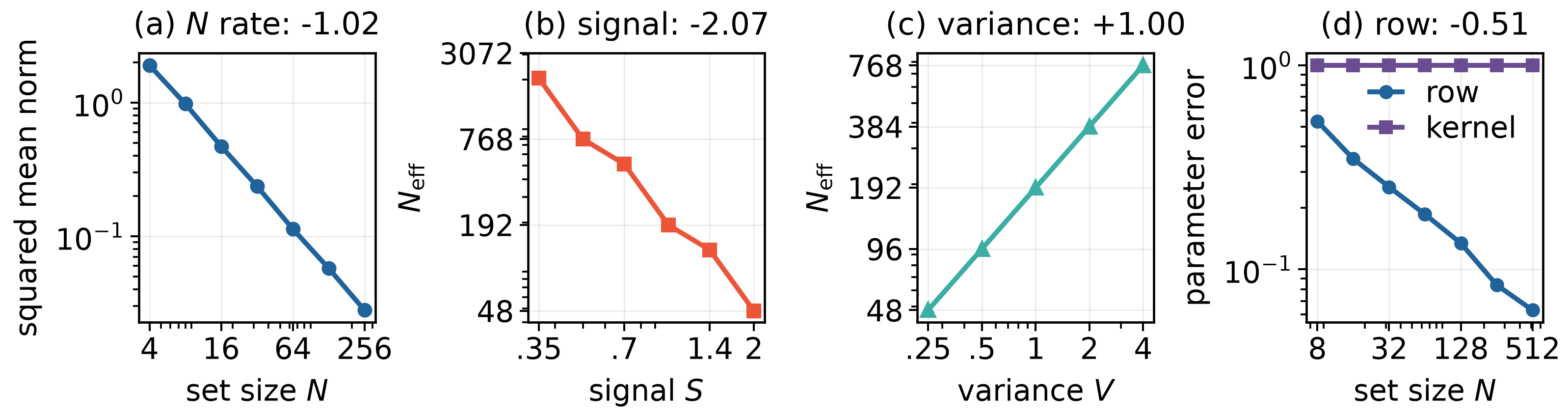}
\caption{Controlled validation of moment visibility and Gaussian scaling. Sampling reduces estimation error in moment-sensed directions, stronger shared signals require fewer observations, larger individual variance requires more observations, and an exact Jacobian-kernel component remains unresolved. Slopes are log-log fits.}
\label{fig:sup-theory}
\end{figure*}

\paragraph{Moment visibility.}
The linear construction used $y=x+G\theta^*$ with $G\in\R^{6\times p}$ and $C_\theta(y)=y-G\theta$. The statistic was $Q\bar c_\theta$ for $Q\in\R^{k\times6}$, where $\bar c_\theta=N^{-1}\sum_i C_\theta(y_i)$. Clean data followed $\mathcal N(0,I_6)$ and $J=-QG$. Parameter coordinates were standardized before applying the minimum-norm rule.

With $p=4$, $k=2$, and 400 repetitions, mean row-space error decreased from $0.5314$ to $0.0631$. Its slope was $-0.509$ with $R^2=0.997$, while the injected unit kernel component remained $1.000$. A separate $k=5$ draw had rank four and reduced total error from $0.6878$ to $0.0901$.

\paragraph{Two-sample precision and conditioning.}
The reference-size experiment used seed 20270717, $d=8$, $V=1$, and $N=64$. Each $M$ value used 20,000 exact Gaussian sample-mean errors. Empirical parameter MSE differed from $dV(1/N+1/M)$ by at most $0.58\%$. The curve approached $dV/N$ as $M$ increased.

The scalar interval was
\[
\widehat s_1\pm1.96\sqrt{V(1/N+1/M)}.
\]
Empirical coverage ranged from $94.84\%$ to $95.28\%$ across the tested $N$ and $M$. The weak-coordinate experiment varied $\sigma_{\min}(K_W)$ from 1 to 0.05 with $N=256$. Its log-log slope was $-1.999$ with $R^2=1.000$.

\begin{figure*}[t]
\centering
\includegraphics[width=0.98\textwidth]{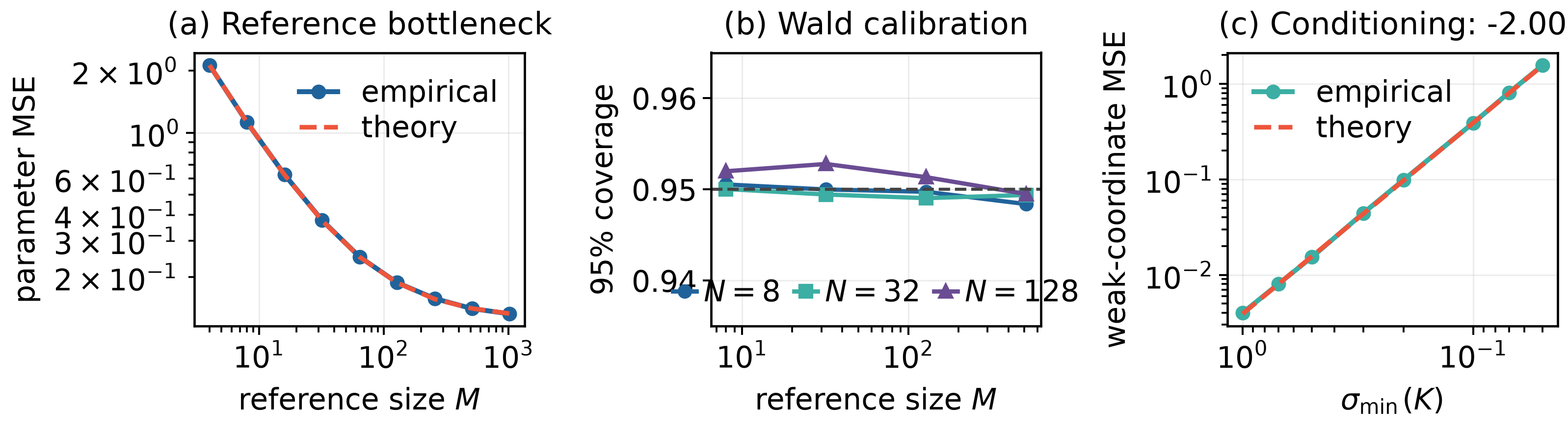}
\caption{Controlled validation of two-sample precision and scaled conditioning. Empirical MSE tracks the separate observed/reference covariance law and approaches the irreducible observed-set floor; exact Gaussian intervals retain nominal coverage; weak-direction MSE follows the predicted $\sigma_{\min}^{-2}$ amplification. Each setting uses 20,000 draws.}
\label{fig:sup-uncertainty}
\end{figure*}

\paragraph{Known-mean Gaussian efficiency.}
A $d=16$ Gaussian location simulation used a known clean mean and 4,000 repetitions. Table~\ref{tab:sup-crb} compares total parameter MSE with $dV/N$.

\begin{table}[t]
\centering
\small
\begin{tabular}{@{}rrrr@{}}
\toprule
$N$ & Empirical MSE & $dV/N$ & Ratio \\
\midrule
4   & 3.9687 & 4.0000 & 0.992 \\
8   & 2.0028 & 2.0000 & 1.001 \\
16  & 0.9862 & 1.0000 & 0.986 \\
32  & 0.4969 & 0.5000 & 0.994 \\
64  & 0.2512 & 0.2500 & 1.005 \\
128 & 0.1260 & 0.1250 & 1.008 \\
\bottomrule
\end{tabular}
\caption{Empirical Gaussian error matches the known-mean Cram\'er--Rao bound across tested set sizes. The simulation uses $d=16$ and $V=1$.}
\label{tab:sup-crb}
\end{table}

\paragraph{Gaussian scaling.}
The set-size, signal-strength, and variance studies used $d=8$, NumPy seed 0, 200 repetitions, and $M=40{,}000$. Squared mean Euclidean error decreased from $1.9011$ to $0.0280$ across the tested $N$ values, with log-log slope $-1.017$. The first grid point below mean relative error $0.20$ defined $N_{\mathrm{eff}}$. Its slopes were $-2.068$ against $S$ and $1.000$ against $V$.

Together, these controlled studies close all five stochastic predictions used by the paper: $N^{-1}$ error scaling in sensed directions, $S^{-2}$ and $V^{+1}$ sample requirements, separate $N^{-1}$/$M^{-1}$ precision contributions, $\sigma_{\min}^{-2}$ weak-direction amplification, and persistence of exact kernel components. The measured exponents and exact-Gaussian checks agree with the corresponding analytic quantities without extending those finite-sample claims to the nonlinear operator studies.

\subsection{Color Corrected-Output Calibration}

The color instance resized images to $256\times256$ and used
\[
C_\theta=G_\eta\circ S_s\circ\Gamma_\gamma\circ R_\beta\circ A_\alpha.
\]
The affine block used 12 coordinates. The monotone piecewise-linear curve used 16 positive increments per channel, while gamma, saturation, and gain-bias used 3, 1, and 6 coordinates. Thus $p=12+48+3+1+6=70$. The same stages acted on every image and clamped outputs to $[0,1]$.

Means and standard deviations were pooled over batch and spatial positions with variance denominator $L-1$. RGB channels contributed six coordinates. The 256-channel VGG-16 \texttt{features[:16]} output contributed 512, giving $q=518$. VGG-16 used ImageNet-1K V1 weights in evaluation mode and standard normalization. RGB summaries were multiplied by five before squared-error matching with $W=I$. Parameters were initialized at the identity and optimized for 300 Adam steps with learning rate $0.05$. Reference summaries used 200 clean faces.

The source images were unretouched \texttt{bare.jpg} files extracted from FFHQ \cite{karras2019stylegan}. The data-generation script resized them to $256\times256$ and applied nonlearned Instagram-style filter formulas implemented by \texttt{pilgram}. Identities were sorted before splitting: positions 0--1999 formed the training pool and positions 2000--2299 the held-out test pool, with measured intersection zero. The paired training manifest used 1800 identities and 16 filters; the validation manifest used the remaining 200 training-pool identities and the same 16 filters. The nine test filters---aden, brooklyn, inkwell, lark, maven, moon, rise, slumber, and stinson---were disjoint from those 16. Because the two sets partition all 25 filters in \texttt{pilgram}, ``unseen'' here means held out within one implementation family, not a filter from an independent app or engine.

For each test filter, CSV rows \texttt{[:N]} supplied filtered observed images and rows \texttt{[N:N+40]} supplied queries. Clean counterparts of the observed images were never used in fitting; query clean images served only as evaluation ground truth. Thus the main $N=30$ comparison used rows 0--29 for calibration and 30--69 for 40 queries per filter, or 360 queries total. The set-size sweep used the same subsequent-40 rule for each $N$. The clean reference consisted of 200 unpaired FFHQ \texttt{bare.jpg} entries at half-open sorted indices \texttt{[15000:15200]}, outside the identity range used for training, validation, calibration, or queries. LPIPS decreased from $0.1788$ at $N=1$ to $0.1020$, $0.0610$, and $0.0578$ at $N=3,10,30$. At $N=30$, replacing the shared-filter calibration set with a set assembled from other filters increased LPIPS from $0.0578$ to $0.1150$.

The fixed supervised row used the seed-1 ARColorOpV2 best checkpoint, not NAFNet. It was trained on paired data from the 16 seen filters and selected at epoch 39 with validation LPIPS 0.0330. The separately available NAFNet checkpoint was not used for this row.

A query counted as harm when its corrected LPIPS exceeded that of its filtered input. Table~\ref{tab:sup-operators} records the information-regime comparison and the channel endpoints.

\begin{table}[t]
\centering
\small
\begin{tabular}{@{}llrr@{}}
\toprule
Instance & Information / method & Metric $\downarrow$ & Harm \\
\midrule
Color & none / no operation & 0.0954 & 0\% \\
 & paired training / ARColorOpV2 & 0.1023 & 70\% \\
 & unpaired set / Reinhard & 0.0761 & 17\% \\
 & unpaired set / moment, $N=30$ & 0.0547 & 1\% \\
\midrule
Channel & $N=1$ & 0.2340 & -- \\
 & $N=256$ & 0.0012 & -- \\
\bottomrule
\end{tabular}
\caption{Color rows report LPIPS and harm frequency; channel rows report relative magnitude-response MSE. The rows use different information regimes and are not a common-training leaderboard.}
\label{tab:sup-operators}
\end{table}

\begin{figure*}[t]
\centering
\includegraphics[width=0.98\textwidth]{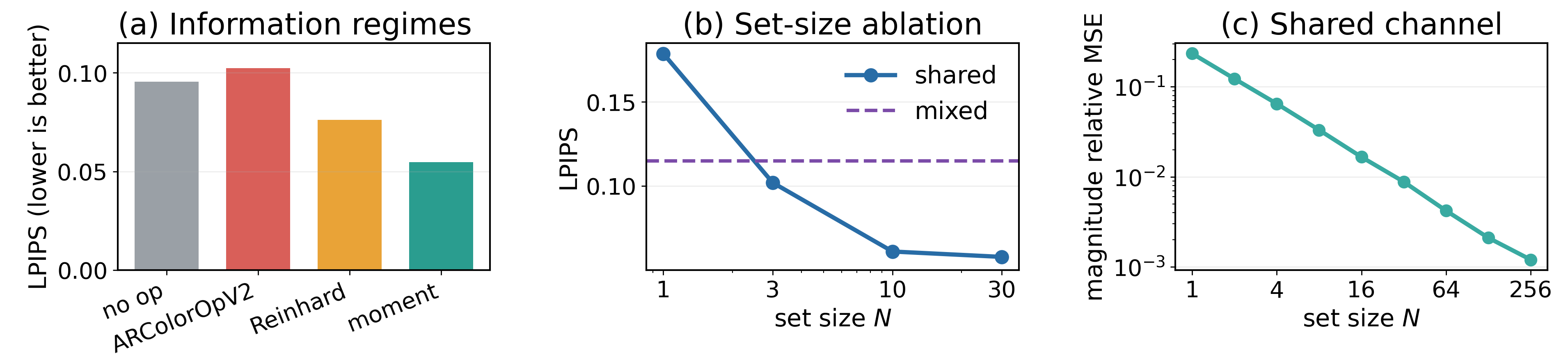}
\caption{Bounded operator results. (a) At $N=30$, moment calibration attains the lowest LPIPS among the reported information regimes. (b) Shared-filter calibration improves steadily with set size, while replacing it with a mixed-filter set at $N=30$ degrades LPIPS from $0.0578$ to $0.1150$, confirming that the set must share the corruption being calibrated. (c) Channel magnitude-response error decreases nearly as $N^{-1}$.}
\label{fig:sup-operators}
\end{figure*}

\subsection{Paired Geometric De-Beautification}

The image-level geometry application used approximately 480 paired $256\times256$ FFHQ faces processed by one commercial face-reshape operator. Sorted pairs were split 85\%/15\% into training and test sets, leaving 72 test images. The proposed width-48 U-Net has two output heads. One predicts a stationary velocity field, integrated with six scaling-and-squaring steps before differentiable warping; the other predicts an appearance residual added to the warped input. It was trained for 5,000 Adam updates with batch size 16 and learning rate $2\times10^{-4}$. Its loss was LPIPS plus $0.7$ times pixel L1, $0.2$ times identity loss, $0.01$ times mean squared velocity, $2.0$ times the negative-Jacobian penalty, and $0.001$ times mean squared residual. The width-matched Direct U-Net used the same pairs and split, predicted an input residual directly, and was trained for 4,000 updates with LPIPS plus $0.7$ times pixel L1. The paper-facing proposed-method label is \emph{Ours}.

On the in-distribution test set, Ours achieved PSNR $39.83$, SSIM $0.972$, LPIPS $0.0049$, and identity similarity $0.997$. Direct U-Net achieved $38.34$, $0.972$, $0.0061$, and $0.995$, respectively. Figure~\ref{fig:sup-debeautification} shows all four retained qualitative examples and their per-image PSNR/LPIPS values. This experiment uses paired supervision and a per-image nonrigid model. It therefore supports only the stated de-beautification application and is not empirical validation of the unpaired moment estimator or Proposition~2.

\subsection{Channel Magnitude-Response Calibration}

The channel instance drew length-256 AR(1) signals and passed them through one fixed normalized eight-tap LTI channel. The reference estimate used 4,000 clean signals. On evaluated frequencies, the estimator was
\[
\widehat{|H(\omega)|}
=\left[
\max\left\{
\frac{\widehat P_Y(\omega)}
{\widehat P_X(\omega)+10^{-9}},
10^{-8}
\right\}
\right]^{1/2}.
\]
Same-mode convolution and AR initialization introduced finite-length boundary effects. The instance estimates channel magnitude while leaving phase unresolved.

Averaged over 30 trials, relative magnitude-response MSE was $0.2340$, $0.1226$, $0.0648$, and $0.0331$ for $N=1,2,4,8$. It was $0.0166$, $0.0088$, $0.0042$, $0.0021$, and $0.0012$ for $N=16,32,64,128,256$. For $N\ge2$, the log-log slope was $-0.971$ with $R^2=1.000$.

\section{End-to-End Interpretation and Reproducibility}

\subsection{Closed Evidence Chain}

The prescribed correction and fixed moments define $g(\theta)$; Jacobian rank determines local visibility, the projected influence expansion separates the two sampling sources, and $\sigma_{\min}(K_W)$ quantifies the weakest estimable direction. The controlled studies recover these predicted laws. The shared-set color ablation directly tests the central premise, the channel study shows stable magnitude recovery over a 256-fold set-size range, and the separate paired geometry model improves every aggregate metric except tied SSIM. Thus the full chain closes at the level supported by each experiment without conflating unpaired calibration with paired nonrigid restoration.

\subsection{Auditable Protocol Record}

The row-kernel, Gaussian-scaling, and efficiency studies use 400, 200, and 4,000 repetitions; every two-sample grid point uses 20,000 draws. Seeds are 20270717 for reference size and 0 for Gaussian scaling. For color, training, test, and reference identities are disjoint; observed clean counterparts never enter fitting, and query clean images are ground truth only. The nine evaluation filters are disjoint from the 16 learned-baseline filters. The fixed learned row is ARColorOpV2, not NAFNet. Geometry and channel splits, runs, hyperparameters, estimands, and boundaries are stated in their respective subsections.

\bibliography{setlevel2027}

\clearpage
\begin{figure*}[p]
\centering
\includegraphics[width=0.96\textwidth]{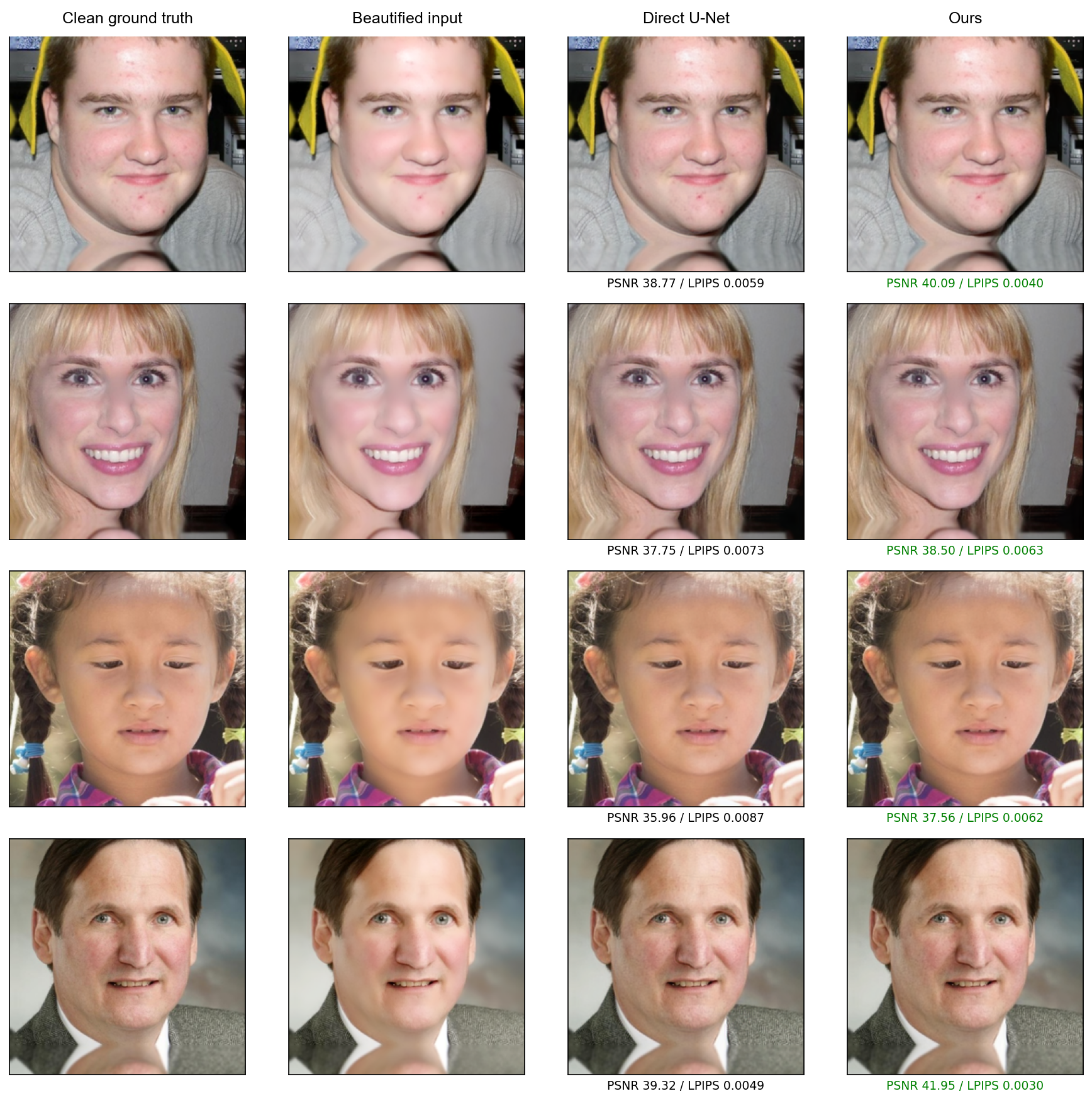}
\caption{Complete qualitative comparison for paired geometric de-beautification. Columns show clean ground truth, beautified input, the task-specific Direct U-Net, and Ours. The printed PSNR/LPIPS values favor Ours in all four displayed examples; over all 72 test images, Ours reaches PSNR $39.83$ and LPIPS $0.0049$ versus $38.34$ and $0.0061$ for Direct U-Net.}
\label{fig:sup-debeautification}
\end{figure*}
\clearpage